\documentclass[12pt]{article}
\usepackage{graphicx}
\usepackage{amsfonts}
\usepackage{amsmath}
\usepackage{amssymb}
\usepackage{url}
\usepackage{tikz}
\usepackage{pgfplots}
\pgfplotsset{compat=1.17}
\usepackage{dsfont} 

\usepackage{indentfirst}
\usepackage{enumerate}
\usepackage[normalem]{ulem}
\usepackage{mathrsfs}
\usepackage{multirow}
\usepackage{diagbox}
\usepackage{bbold}
\usepackage{stackengine}
\usepackage{subfigure}
\renewcommand{\thesubfigure}{\alph{subfigure}}
\makeatletter
\renewcommand{\@thesubfigure}{(\thesubfigure)\hskip\subfiglabelskip}
\makeatother
\usepackage[colorlinks=true,citecolor=blue]{hyperref}
\usepackage{amsthm}
\usepackage{color}

\definecolor{DarkGreen}{rgb}{0.2,0.6,0.2}

\definecolor{purple}{rgb}{0.6,0.3,0.8}

\usepackage{natbib} 
\usepackage{comment}
\allowdisplaybreaks

\def\d{\mathrm{d}}

\newcommand{\E}{\mathbb{E}}

\newcommand{\R}{\mathbb{R}}

\newcommand{\dsquare}{\mathop{  \square} \displaylimits}

\newcommand{\p}{\mathbb{P}}
\newcommand{\X}{\mathcal{X}}

\renewcommand{\H}{\mathcal{H}}

\newcommand{\id}{\mathds{1}}

\renewcommand{\ge}{\geqslant}
\renewcommand{\le}{\leqslant}
\renewcommand{\geq}{\geqslant}
\renewcommand{\leq}{\leqslant}
\renewcommand{\epsilon}{\varepsilon}
\newcommand{\esssup}{\mathrm{ess\mbox{-}sup}}
\newcommand{\essinf}{\mathrm{ess\mbox{-}inf}}
\renewcommand{\cdots}{\dots}
\theoremstyle{plain}
\newtheorem{theorem}{Theorem}
\newtheorem{corollary}{Corollary}
\newtheorem{lemma}{Lemma}
\newtheorem{proposition}{Proposition}
\theoremstyle{definition}
\newtheorem{definition}{Definition}
\newtheorem{example}{Example}

\theoremstyle{remark}

\DeclareMathOperator*{\argmin}{arg\,min}

\newcommand{\dboxminus}{\mathop{  \boxminus} \displaylimits}
\newcommand{\dboxplus}{\mathop{  \boxplus} \displaylimits}
\newcommand{\lozen}{\mathop{  \lozenge} \displaylimits}

\usepackage{setspace}

\usepackage{footmisc}
\title{Counter-monotonic Risk Sharing with\\ Heterogeneous Distortion Risk Measures\vspace{0.2cm}}

\author{
Mario Ghossoub\thanks%
  {Department of Statistics and Actuarial Science,
  University of Waterloo,
  Waterloo, Ontario, Canada.
  E-mail: \href{mailto:mario.ghossoub@uwaterloo.ca}{
mario.ghossoub@uwaterloo.ca}.}
  \and Qinghua Ren\thanks%
  {Department of Statistics and Actuarial Science,
  University of Waterloo,
  Waterloo, Ontario, Canada.
  E-mail: \href{mailto:qinghua.ren@uwaterloo.ca}{qinghua.ren@uwaterloo.ca}.}
  \and Ruodu Wang\thanks%
  {Department of Statistics and Actuarial Science,
  University of Waterloo,
  Waterloo, Ontario, Canada.
  E-mail: \href{mailto:wang@uwaterloo.ca}{wang@uwaterloo.ca}.}
  }

\begin{document}
\maketitle

\begin{abstract}
We study  risk sharing among agents with preferences modeled by heterogeneous distortion risk measures, who are not necessarily risk averse. Pareto optimality for agents using risk measures is often studied through the lens of inf-convolutions, because allocations that attain the inf-convolution are Pareto optimal, and the converse holds true under translation invariance.
Our main focus is on groups of agents who   exhibit varying levels of risk seeking. Under mild assumptions, we derive explicit solutions for the unconstrained inf-convolution and the counter-monotonic inf-convolution, which can be represented by a generalization of distortion risk measures. 

\textbf{Keywords:} Pareto optimality; risk sharing; counter-monotonic improvement; risk seeking; inf-convolution
\end{abstract}


\section{Introduction}
 
Risk-exchange markets, such as insurance, reinsurance, or financial markets,  are central  to modern economics.  The primary focus of studying such markets has traditionally been the determination of an optimal, or efficient redistribution of the aggregate market risk, through contracts or trading mechanisms,  among market participants, henceforth referred to as agents. The seminal work of \cite{borch1962} and \cite{wilson1968theory} showed that within the framework of Expected-Utility Theory (EUT), Pareto-optimal allocations between risk-averse agents are comonotonic, and they can therefore be expressed as nondecreasing functions of the aggregate market risk. This is a cornerstone result in the theory of risk sharing, and it is often seen as a foundational justification for risk pooling, since each agent's risk allocation at an optimum depends only on the realization of the aggregate risk. Numerous extensions beyond EUT have been studied in the literature, with the perennial assumption of (strong) risk aversion, that is, monotonicity with respect to the concave order. The literature is vast, and we refer for instance to the work of \cite{chateauneuf2000optimal}, \cite{Dana2002, dana2004ambiguity}, \cite{tsanakas2006risk}, \cite{decastro2011ambiguity}, \cite{beissner2023optimal}, and \cite{RavanelliSvindland2014} for several models of ambiguity-sensitive preferences.

A milestone result in this direction is the so-called \textit{comonotonic improvement theorem} (e.g., \cite{landsberger1994co}, \cite{dana2004ambiguity}, \cite{ludkovski2008comonotonicity}, \cite{carlier2012pareto}, or \cite{denuitetal2023comonotonicity}), an important implication of which is that risk-averse agents always prefer comonotone allocations, and that Pareto optima are comonotonic under strict risk aversion. This naturally led \cite{boonen2021competitive} to examine the so-called comonotone market, 
an incomplete market in which only comonotonic allocations are feasible. Pareto-optimal risk sharing in comonotone markets was also studied by \cite{liu2020weighted} and \cite{ghossoubzhu2024}.

In the risk measurement literature, Pareto-optimal risk sharing between agents with convex or coherent risk measures has been widely studied as well. We refer to \cite{barrieu2005inf}, \cite{acciaio2007optimal}, \cite{jouini2008optimal}, \cite{filipovic2008optimal},  \cite{MastrogiacomoRosazza}, and the references therein, for instance. Additionally, Pareto-optimal risk sharing between agents with quantile-based risk measures that are not necessarily convex was examined by \cite{embrechts2018quantile,embrechts2020quantile}, \cite{liu2020weighted}, \cite{liebrich2024risk} and \cite{ghossoubzhuchong2024}, for instance.

The characterization of optimal allocations in risk sharing markets involving agents who are not risk-averse remains relatively underexplored. 
Recent studies on quantile-based risk sharing, including \cite{embrechts2018quantile, embrechts2020quantile} and \cite{weber2018solvency}, identified a pairwise counter-monotonic structure --- the opposite of comonotonicity --- in an optimal allocation. Furthermore, \cite{lauzier2023risk} provided explicit Pareto-optimal allocations among agents using the inter-quantile difference, demonstrating that an optimal allocation exhibits a mixture of pairwise counter-monotonic structures. 
As a dependence concept, pairwise counter-monotonicity has been studied by \cite{dall1972frechet}, \cite{dhaene1999safest} and \cite{cheung2014characterizing}. 
Parallel to the comonotone improvement theorem, \cite{lauzier2024negatively} established the so-called \textit{counter-monotonic improvement theorem}, leading to an implication that counter-monotonic allocations will always be preferred by risk-seeking agents. Based on the counter-monotonic improvement theorem, \cite{ghossoub2024counter} provided a systematic study of risk sharing in markets where only counter-monotonic allocations are allowed, and they gave an explicit characterization of optimal allocations when agents are risk-averse and risk-seeking.  Their analysis assumes that  the preferences of the agents are modelled by a common distortion risk measure.

This paper extends the previous work by examining a market where agents may have heterogeneous risk preferences, and they are not necessarily risk averse. It is notable that although agents within a group may differ in their levels of risk aversion, we assume that all agents in the same group are either all risk-averse or all risk-seeking. We do not consider cases where both risk-averse and risk-seeking agents are combined in a single group. The key novel insights and extensions are summarized as follows: 
\begin{itemize}
    \item[(i)] In the homogeneous case, \cite{ghossoub2024counter} established a universal ordering among the three versions of inf-convolution: unconstrained, counter-monotonic, and comonotonic, from the smallest to the largest. In contrast to the homogeneous case, the ordering between the comonotonic and counter-monotonic versions of the inf-convolution depends on the distortion functions. 
    \item[(ii)] When agents have identical concave distortion functions (the case of risk aversion), the three versions of inf-convolution have identical values, as shown in Theorem 3 of \cite{ghossoub2024counter}. 
    However, for a group of risk-averse agents with different levels of risk aversion, counter-monotonic allocations are generally not Pareto optimal, leading to a gap between the three versions of inf-convolution.  
    \item[(iii)] Under some mild conditions, we derive an explicit formula for the counter-monotonic inf-convolution in the case where agents are risk seeking, characterized by different convex distortion functions.  
\end{itemize}

 The rest of the paper is organized as follows. Sections \ref{sec:pre} and  \ref{sec:counter} contain preliminaries on risk measures and on risk sharing problems, respectively. In particular, Section \ref{sec:counter} recalls the counter-monotonic improvement theorem (reported as Theorem \ref{theorem:counter_impro}) and some related discussions on counter-monotonicity. In Section  \ref{sec:concave}, we analyze the counter-monotonic risk sharing problem and obtain general relations for different choices of the distortion functions (Theorem \ref{thm:1}).
 We specialize to risk-seeking agents in  Section \ref{sec:risk-seeking}. Based on the counter-monotonic improvement theorem, counter-monotonic inf-convolutions are determined explicitly for risk-seeking agents (Theorem \ref{thm:convex}). 
Section \ref{sec:con} concludes the paper. 

\section{Preliminaries}\label{sec:pre}

\subsection{Risk measures and basic terminology}
 Let $(\Omega, \mathcal{F}, \mathbb{P})$ be an atomless probability space and $\mathcal{X}$ a convex cone of random variables on this space.  Section \ref{sec:concave} considers $\mathcal{X}=L^\infty$, the set of all essentially bounded random variables, and Section  \ref{sec:risk-seeking}   considers $\mathcal{X}=L^+$ or $\mathcal{X}=L^-$, where $L^+$ (resp.~$L^-$) represents the sets of nonnegative (resp.~nonpositive) essentially bounded random variables. 
 Almost surely equal random variables are treated as identical. 
Throughout, the random variable $X\in \X$ represents losses, and its negative values represent gains.  
 We denote by $\mathbb{1}_A$  the indicator function for an event $A \in \mathcal{F}$. 
 Let $$\H^{\rm BV}=\{h:  [0,1]\to \R \mid  h \text{ is of bounded variation and } h(0)=0 \}.$$
 Next, we present the definition of a distortion riskmetric.

 \begin{definition}
 A \emph{distortion riskmetric} is a mapping $\rho_h: \mathcal{X} \to \mathbb{R}$ given by 
 \begin{align*}
\rho_h(X)=\int X \d\left( h\circ \p \right)=\int_0^\infty h(\p(X \geq  x))\, \d x + \int_{-\infty}^{0} (h(\p(X \geq  x))-h(1) )\,\d x, 
\end{align*} 
for some $h\in \H^{\rm BV}$. 
\end{definition}

We note that the elements of $\H^{\rm BV}$ are not necessarily monotone. If we constrain $h \in \H^{\rm BV}$ to be increasing and normalized, that is, $h\in \H$, where 
$$\H=\{h:  [0,1]\to [0,1]\mid  h \,\text{is increasing and } h(0)=1-h(1)=0 \},$$ 
then the distortion riskmetric $\rho_h$ for $h\in\H $ is a distortion risk measure. 
Here and throughout, terms like ``increasing" or ``decreasing" are in the non-strict sense.  
In this paper, the agents' risk preferences are modeled by the class of distortion risk measures. The more general class of distortion riskmetrics is introduced since it will be useful in our further analysis. In particular, we show in the setting of Section \ref{sec:risk-seeking} that the comonotonic inf-convolution of distortion risk measures is a distortion risk measure, whereas their counter-monotonic inf-convolution is a distortion riskmetric.

The \textit{dual} $\tilde h$ of a given $h\in \mathcal H^{\rm BV}$, which will be useful in many of our results, is defined as 
$$
\tilde h (t) = h(1) - h(1-t), \ \mbox{~~~~for all  $t \in [0,1]$,}
$$
and it is an element of $\mathcal H^{\rm BV}$.
If $h$ is in $\mathcal H$, then so is  $\tilde h$. The dual of $\tilde h$ is equal to $h$, and the two corresponding distortion riskmetrics are connected via the equality 
$$\rho_h(X) = - \rho_{\tilde h}(-X), \mbox{~~~~for all  $X\in \X$.}$$

We now recall some properties of distortion riskmetrics that we use throughout. A distortion riskmetric $\rho_h$ may have the following   properties as a functional $\rho: \mathcal{X}\mapsto \mathbb{R}$.
\begin{enumerate}[(a)]
    \item Law-invariance:  $\rho(X)=\rho(Y)$ if $X$ and $Y$ have the same distribution, i.e., $X \stackrel{\mathrm{d}}{=} Y$; 
    \item  Positive homogeneity: $\rho(\lambda X)=\lambda \rho(X)$ for any $\lambda>0$;
    \item Translation invariance: $\rho (X + c) = \rho (X) +c$ for $c\in \mathbb{R}$ and $X+c\in \X$;
    \item Comonotonic additivity: $\rho(X+Y)=\rho(X)+\rho(Y)$ if $X$ and $Y$ are comonotonic;
    \item Subadditivity: $\rho(X+Y) \leqslant \rho(X)+\rho(Y)$;
    \item Convex order consistency: $\rho(X)\le \rho(Y)$ if $X\le_{\rm cx} Y$, where the inequality is the convex order, meaning $\E[u(X)]\le \E[u(Y)]$ for all convex functions $u$ such that the two expectations are well-defined;
    \item Monotonicity: $\rho(X) \le \rho(Y)$ if $X\le Y$.
\end{enumerate} 
In fact, these properties do not always hold for a distortion riskmetric $\rho_h$.  To be more specific, all distortion riskmetric $\rho_h$ satisfy (a), (b), and (d). 
Property (c) holds true if $h(1)=1$.
By \citet[Theorem 3]{wang2020characterization}, conditions (e) and (f) are equivalent to the concavity of $h$. 
Condition (g) is equivalent to increasing monotonicity of $h$. The four properties (b), (c), (e) and (g) together define a coherent risk measure in the sense of \cite{artzner1999coherent}, corresponding to an increasing and concave $h$ with $h(1)=1$ for $\rho_h$. 
For various characterizations and properties of distortion riskmetrics, see   \cite{wang2020characterization} on $L^\infty$ and \cite{wang2020distortion}  on more general spaces.
Convex order in (f) and its related notions are popular for modeling risk aversion in decision theory \citep{RothschildStiglitz1970}, and it is also widely studied in actuarial science and risk management \citep{ruschendorf2013mathematical, he2016risk}.

Many popular risk measures belong to the family of distortion risk measures, including the  regulatory risk measures used in the banking and insurance sectors, namely, the Value-at-Risk (VaR) and the Expected Shortfall (ES, also known as CVaR and TVaR), which are defined below. 
For a random variable $X$, the VaR at level $\alpha \in  [0, 1)$ is defined as 
\begin{align}\label{eq:varrr}
    \operatorname{VaR}_\alpha(X)=\inf \{x \in[-\infty, \infty]: \mathbb{P}(X \leqslant x) \geqslant 1-\alpha\} ,
\end{align}
and  the ES  at level $\beta \in  [0, 1)$ is defined as
\begin{align*}
    \operatorname{ES}_\beta(X)=\frac{1}{\beta} \int_0^\beta \operatorname{VaR}_\gamma(X) \mathrm{d} \gamma,
\end{align*}
where $\operatorname{VaR}_\gamma$ is defined in (\ref{eq:varrr}).
Here we use the convention of ``small $\alpha$" as in \cite{embrechts2018quantile}.
If $\alpha \in [0,1)$, $\operatorname{VaR}_{\alpha}$ and $\operatorname{ES}_{\alpha}$ 
are distortion risk measures, and they 
are associated with the distortion functions $h(t)=\id_{\{t>\alpha\}}$ and $ h(t)=\min \{t/\alpha, 1\}$, respectively.

In this paper, we write $X \sim F_{X}$ for $X \in \mathcal{X}$ having cumulative distribution $F_X$ and survival distribution $S_X$. Since the space $(\Omega, \mathcal{F}, \mathbb{P})$ is atomless, for each $X \in \mathcal{X}$ there exists a random variable $U_X$ with a uniform distribution on $[0,1]$ such that $F_X^{-1}(U_X)=X$, almost surely. The existence of such a $U_X$ for any random variable $X$ follows from \citet[Lemma A.32]{follmer2016stochastic}. For $x, y \in \mathbb{R}$, write $ x \vee y=\max \{x, y\}$ and $x \wedge y=\min \{x, y\}$.

\subsection{Risk sharing and inf-convolution}
\noindent 
Given $X \in \mathcal{X}$, the set of feasible allocations of $X$ is defined as
\begin{align}\label{def:allo}
    \mathbb{A}_n(X)=\left\{\left(X_1, \dots, X_n\right) \in \mathcal{X}^n: \sum_{i=1}^n X_i=X\right\}.
\end{align}

We consider a risk-sharing market, in which $n \in \mathbb{N}$ agents wish to share an aggregate risk $X \in \mathcal{X}$. All throughout, we let $[n]=\{1,\dots,n\}$, and we assume that agent $i \in[n]$ has a risk preference modeled by a risk measure $\rho_{i}$. The market redistributes the aggregate risk $X$ into an allocation $(X_1, \dots, X_n)\in \mathbb{A}_n(X)$, and we refer to $\sum_{i=1}^n \rho_{i}(X_i)$ as the associated aggregate risk value. Note that the definition of allocations depends on the specification of $\X$, which will vary across different applications in the later sections. 

Using \eqref{def:allo}, the inf-convolution $\square_{i=1}^n \rho_{i}$ of $n$  risk measures $\rho_{1}, \dots, \rho_{n}$ is defined as
\begin{align*}
    \dsquare_{i=1}^n \rho_{i}(X):=\inf \left\{\sum_{i=1}^n \rho_{i}\left(X_i\right):\left(X_1, \dots, X_n\right) \in \mathbb{A}_n(X)\right\}, \quad X \in \mathcal{X}.
\end{align*}
That is, the inf-convolution of $n$ risk measures is the infimum over aggregate risk values for all possible allocations.

An allocation $(X_1, \dots, X_n)$ is sum-optimal in $\mathbb{A}_n(X)$ if $\square_{i=1}^n \rho_{i}(X)=\sum_{i=1}^n \rho_{i} (X_i)$, i.e., it attains the optimal total risk value. An allocation $ (X_1, \dots, X_n ) \in \mathbb{A}_n(X)$ is Pareto optimal in $\mathbb{A}_n(X)$ if for any $(Y_1, \dots, Y_n ) \in \mathbb{A}_n(X)$ satisfying $\rho_{i}(Y_i) \leq \rho_{i}(X_i)$ for all $i \in[n]$, we have $\rho_{i}\left(Y_i\right)=\rho_{i}\left(X_i\right)$ for all $i \in[n]$.
Pareto optimality means that the allocation cannot be improved upon for all agents while providing a strict improvement for at least one agent. For distortion risk measures, the equivalence between Pareto optimality and sum-optimality is guaranteed when $\X=L^\infty$, as obtained in \citet[Proposition 1]{embrechts2018quantile}.
We will focus on sum-optimal allocations in this paper. However, although 
sum-optimal allocations are always Pareto optimal, the converse may not hold in case $\X=L^+$ or $X=L^-$, which are examined in Sections \ref{sec:risk-seeking}.

\section{Comonotonic and counter-monotonic risk sharing}
\label{sec:counter}
The elements in the allocation set $\mathbb{A}_n(X)$ can exhibit different dependence structures, with comonotonicity and counter-monotonicity being the two extreme cases. 
We first define comonotonicity and counter-monotonicity for bivariate random variables.
\begin{definition}
    Two random variables $X$ and $Y$ on $(\Omega, \mathcal{F}, \mathbb{P})$ are said to be comonotonic (resp.~counter-monotonic) if
    \begin{align*}
        (X(\omega)-X(\omega^{\prime}))(Y(\omega)-Y(\omega^{\prime})) \ge  0 ~(
        \mbox{resp.}~\le 0),\ \text{for $(\p\times \p)$-almost all}\  (\omega, \omega^{\prime}) \in \Omega^2.
    \end{align*}
\end{definition} 
Note that $X$ and $Y$ are counter-monotonic if and only if $X$ and $-Y$ are comonotonic. Also, comonotonicity of $(X,Y)$ is equivalent to the existence of increasing functions $f$ and $g$ and a random variable $Z$, such that $(X, Y)=(f(Z), g(Z))$ (almost surely), which follows from the stochastic representation of comonotonicity given in \citet[Proposition 4.5]{denneberg1994non}.
Since we treat almost surely identical random variables as equal, we will omit ``almost surely" in statements like the one above. Comonotonicity is foundational to modern ambiguity models in decision theory \citep{schmeidler1989subjective} and  it is widely studied in actuarial science and risk management \citep{dhaene2002concept, dhaene2006risk}.
Counter-monotonicity  also has special roles, quite different from comonotonicity, in decision theory \citep{principi2023antimonotonicity} and actuarial science \citep{cheung2014reducing, chaoubi2020sums}.

Next, we define these concepts in higher dimensions. A random vector $(X_{1},\ldots, X_{n})$ is (pairwise) comonotonic (resp.~counter-monotonic) if each pair of its components is comonotonic (resp.~counter-monotonic). Pairwise counter-monotonicity is the generalization of counter-monotonicity to the case $n \geq 3$, and we will hereafter use the simpler term counter-monotonicity throughout.
Although comonotonicity for $n\ge 3$ is a straightforward extension of the case $n=2$,  counter-monotonicity for $n\ge 3$ imposes strong constraints on the marginal distributions, and it is quite different from the case $n=2$ \citep{dall1972frechet, dhaene1999safest, cheung2014characterizing}. Below we provide some facts about counter-monotonicity. Let $\Pi_n$ be the set of all $n$-compositions of $\Omega$, that is,
$$
\Pi_n=\left\{\left(A_1, \ldots, A_n\right) \in \mathcal{F}^n: \bigcup_{i \in[n]} A_i=\Omega \text { and } A_1, \ldots, A_n \text { are disjoint}\right\}.
$$
Formally, compositions are partitions in which the order matters. 

  We quote below a stochastic representation of counter-monotonicity, which will be useful throughout our analysis. In what follows, $\esssup X$ and $\essinf X$  are the essential supremum and the essential infimum of $X$, respectively. Additionally, a random variable $X \in \mathcal{X}$ is said to be degenerate if there exists $c \in \mathbb{R}$ such that $X=c$ (almost surely).

\begin{proposition}[\cite{lauzier2023pairwise}]\label{prop:counter}
    For $X \in \mathcal{X}$ and $n\ge 3$, suppose that at least three components of $\left(X_1, \ldots, X_n\right) \in \mathbb{A}_n(X)$ are non-degenerate. Then $\left(X_1, \ldots, X_n\right)$ is counter-monotonic if and only if there exist constants $m_1, \ldots, m_n$ and  $\left(A_1, \ldots, A_n\right) \in \Pi_n$ such that
\begin{align}\label{eq:counter_form1}
\mbox{either~~~~} &    X_i=(X-m) \mathbb{1}_{A_i}+m_i  \text { for all } i \in[n], \text { with } m=\sum_{i=1}^n m_i \leq \essinf X, \\
\mbox{or~~~~} \label{eq:counter_form2}
    & X_i=(X-m) \mathbb{1}_{A_i}+m_i \text { for all } i \in[n], \text { with } m=\sum_{i=1}^n m_i \geq \esssup X.
\end{align}
\end{proposition}

By taking $m_{1}=\ldots=m_{n}=0$, which is possible when $X \geq 0$ or $X \leq 0$, a simple counter-monotonic allocation in the form of \eqref{eq:counter_form1} and \eqref{eq:counter_form2} is given by
\begin{align*}
    X_i=X \mathbb{1}_{A_i}, \, \text{for each} \, i\in [n] , \, \text{where}\, \left(A_1, \ldots, A_n\right) \in \Pi_n.
\end{align*}
Specifically, such an allocation is called a \textit{jackpot allocation} if $X \geq 0$ and a \textit{scapegoat allocation} if $X \leq 0$ by \cite{lauzier2024negatively}.
It is clear that there is a ``winner-takes-all" structure in a jackpot allocation and a ``loser-loses-all" structure in a scapegoat allocation.

For a given $X$, we denote by $ \mathbb{A}_n^{+}(X)$ the  set of comonotonic allocations and by $\mathbb{A}_n^{-}(X)$  the set of counter-monotonic allocations, introduced by \cite{ghossoub2024counter}.
These two sets of allocations impose restrictions on the dependence structures of allocations, and both  are strict subsets of the set   $\mathbb{A}_{n}(X)$ of
all possible allocations.

It is well known that for any $(X_1,\dots,X_n)\in \mathbb A_n(X)$, there exists $(Y_1,\dots,Y_n)\in \mathbb A^+_n(X)$ 
such that $Y_i\le_{\rm cx} X_i$ for each $i\in[n]$. 
This is known as 
the comonotonic improvement theorem (e.g., \cite{landsberger1994co}, \cite{ruschendorf2013mathematical}, or \cite{denuitetal2023comonotonicity}).
Recently, \cite{lauzier2024negatively} provided a counter-monotonic improvement  theorem, which states that under mild conditions, for any allocation $(X_1,\dots,X_n)\in \mathbb A_n(X)$, there exists $(Y_1,\dots,Y_n)\in \mathbb A^-_n(X)$ 
such that $Y_i\ge_{\rm cx} X_i$, for each $i\in[n]$. The formal statement is summarized below, and it will be useful for the results of this paper.

\medskip
 
\begin{theorem}[\cite{lauzier2024negatively}]\label{theorem:counter_impro}
     Let $X_1, \ldots, X_n \in L^1$ be nonnegative and $X=\sum_{i=1}^n X_i$. Assume that there exists a uniform random variable $U$ independent of $X$. Then, there exists $\left(Y_1, \ldots, Y_n\right) \in \mathbb{A}_n(X)$ such that (i) $\left(Y_1, \ldots, Y_n\right)$ is counter-monotonic; (ii) $Y_i \geq_{\mathrm{cx}} X_i$, for each $i \in[n]$; and (iii) $Y_1, \ldots, Y_n$ are nonnegative. Moreover, $\left(Y_1, \ldots, Y_n\right)$ can be chosen as a jackpot allocation.
\end{theorem}

The counter-monotonic improvement theorem indicates that jackpot allocations are always preferred by risk-seeking agents. To apply the counter-monotonic improvement theorem,  the technical assumption that there exists a (nondegenerate) uniform random variable $U$  independent of $X$ is used to generate a random lottery, which has utility for risk-seeking agents. To formalize this, we introduce the following set:
    $$\mathcal{X}^{\perp}=\{X \in \mathcal{X}: \text{there exists a uniform random variable $U$ independent of} \,X\}.$$

Due to the comonotonic improvement theorem and the counter-monotonic improvement theorem, Pareto-optimal risk allocations can be constrained in the sets of comonotonic or counter-monotonic allocations, for agents with suitable risk attitudes. Hereafter, we consider risk-sharing problems constrained to these specific allocation structures. The comonotonic inf-convolution $\boxplus_{i=1}^n \rho_{i}$ of risk measures $\rho_{1}, \dots, \rho_{n}$ is defined as
$$
\underset{i=1}{\stackrel{n}{\boxplus}} \rho_{i}(X)=\inf \left\{\sum_{i=1}^n \rho_{i}\left(X_i\right):\left(X_1, \ldots, X_n\right) \in \mathbb{A}_n^{+}(X)\right\}.
$$
Similarly, the counter-monotonic inf-convolution $ \dboxminus_{i=1}^{n} \rho_{i}$ is defined as 
\begin{align*}
    \dboxminus_{i=1}^{n} \rho_{i}(X) = \inf \left\{\sum_{i=1}^{n}\rho_{i}(X_{i}): (X_{1}, \ldots, X_{n}) \in \mathbb{A}_{n}^{-}(X) \right\}.
\end{align*} 
An allocation $\left(X_1, \ldots, X_n\right)$ in $\mathbb{A}_n^{+}(X)$
(resp.~$\mathbb{A}_n^{-}(X)$) is called an optimal allocation of $X$ for $\left(\rho_{1}, \ldots, \rho_{n}\right)$ within $ \mathbb{A}_n^{+}$  (resp.~$ \mathbb{A}_n^{-}$)
if $$\sum_{i=1}^n \rho_{i}\left(X_i\right)=\dboxplus_{i=1}^n \rho_{i}(X) ~~~~~ \left(\mbox{resp.~}\sum_{i=1}^n \rho_{i}\left(X_i\right)=\dboxminus_{i=1}^n \rho_{i}(X)\right).
$$
By definition,   $\square_{i=1}^n \rho_{i}(X) \leqslant \boxplus_{i=1}^{n}\rho_{i}(X)$ and $\square_{i=1}^n \rho_{i}(X) \leqslant \dboxminus_{i=1}^{n}\rho_{i}(X)$.
Hence, if an optimal allocation of $X$ is comonotonic, then it is also an optimal allocation within $ \mathbb{A}_n^{+}$ and $\square_{i=1}^n \rho_{i}(X)=\boxplus_{i=1}^n \rho_{i}(X)$. A similar implication holds if  an optimal allocation of $X$ is counter-monotonic.
 Unconstrained,  comonotonic, and counter-monotonic risk sharing problems
 correspond to those problems over 
$\mathbb{A}_{n}(X)$, 
$\mathbb{A}_{n}^{+}(X)$, and $\mathbb{A}_{n}^{-}(X)$, respectively.

In the rest of the paper, we will assume that each agent is associated with a distortion risk measure.

\section{General relations}\label{sec:concave}

Before delving into the optimal risk-sharing mechanisms for risk-averse and risk-seeking agents, we first provide an overview of the relationships between three types of inf-convolutions: $\dsquare_{i=1}^{n} \rho_{h_i}$, $\dboxminus_{i=1}^{n}\rho_{h_i}$ and $\dboxplus_{i=1}^{n}\rho_{h_i}$.
 In the homogeneous case, these relationships are straightforward to characterize, and they are outlined in \citet[Theorem 1]{ghossoub2024counter}. However, the situation becomes more complex in the heterogeneous case where agents have different risk preferences. Among these three inf-convolutions,  the unconstrained one is always the smallest, since it works on the largest allocation set, but the relationship between the comonotonic and counter-monotonic inf-convolutions is not always the same. Below, we present several special cases in which the relationship between these three variations can be explicitly determined.

In the following examples, we explore the cases of Value-at-Risk (VaR) and Expected Shortfall (ES) using the results of \cite{embrechts2018quantile}.

\begin{example}[Inf-convolution of VaRs]\label{exp:1}
    By \citet[Corollary 2]{embrechts2018quantile}, for $\alpha_{1}, \ldots, \alpha_{n} \geq 0$ and an integrable random variable $X$, if $\sum_{i=1}^{n}\alpha_{i} < 1$ then
    $$\dsquare_{i=1}^n \operatorname{VaR}_{\alpha_i}(X)=\operatorname{VaR}_{\sum_{i=1}^n \alpha_i}(X),$$ 
    and by \citet[Theorem 2]{embrechts2018quantile}, there exists a counter-monotonic optimal allocation of $X$. This gives $\dboxminus_{i=1}^n \operatorname{VaR}_{\alpha_i} \leq \dsquare_{i=1}^n \operatorname{VaR}_{\alpha_i}$, thereby implying that $\dboxminus_{i=1}^n \operatorname{VaR}_{\alpha_i}\! = \dsquare_{i=1}^n \operatorname{VaR}_{\alpha_i}$, since $\dboxminus_{i=1}^n \operatorname{VaR}_{\alpha_i}\! \geq \dsquare_{i=1}^n \operatorname{VaR}_{\alpha_i}$ always holds. Since, in addition, $\dboxplus_{i=1}^n \operatorname{VaR}_{\alpha_i}\! = \operatorname{VaR}_{\bigvee_{i=1}^{n} \alpha_i}$, we obtain 
    $$
    \dboxminus_{i=1}^n \operatorname{VaR}_{\alpha_i} = \dsquare_{i=1}^n \operatorname{VaR}_{\alpha_i} = \operatorname{VaR}_{\sum_{i=1}^n \alpha_i} \leq \dboxplus_{i=1}^n \operatorname{VaR}_{\alpha_i} = \operatorname{VaR}_{\bigvee_{i=1}^{n} \alpha_i}.
    $$
    Note, in particular, that the inequality above is generally not an equality.
\end{example}

\begin{example}[Inf-convolution of ESs]
    For any $\beta_{1}, \ldots, \beta_{n} \in [0,1)$, it holds that 
    $$
    \dboxminus_{i=1}^n \operatorname{ES}_{\beta_i}=\dsquare_{i=1}^n \operatorname{ES}_{\beta_i}= \dboxplus_{i=1}^n \operatorname{ES}_{\beta_i} = \operatorname{ES}_{\bigvee_{i=1}^{n} \beta_i}.
    $$
    Using the subadditivity of ES, we can easily verify that, if $j$ is the agent with $\beta_j=\bigvee_{i=1}^n \beta_i$,
    then the allocation given by $X_j=X$ and $X_i=0$ for $i\ne j$ is optimal (this is also implied by \citet[Theorem 2]{embrechts2018quantile}). Since this allocation is both comonotonic and counter-monotonic, the three inf-convolutions have the same value here. This is a special case of Corollary \ref{co:1} below.
\end{example}

Within the case of VaR, it follows that $\dboxminus_{i=1}^{n}\rho_{h_i}\leq\dboxplus_{i=1}^{n}\rho_{h_i}$ from Example \ref{exp:1}. However, this is not the only possible relationship.
The following theorem provides conditions under which $\dboxminus_{i=1}^{n}\rho_{h_i}\geq\dboxplus_{i=1}^{n}\rho_{h_i}$, and it identifies sufficient conditions for which the inequality becomes equality. First, recall from \cite{ghossoub2024counter} that a function $h \in \mathcal{H}$ is \emph{dually subadditive} if $h$ is subadditive (i.e., $h(x+y)\leq h(x) + h(y)$ for $x, y\in [0,1]$ with $x+y \leq 1$) and  its dual $\tilde{h}$ is superadditive (i.e., $\tilde{h}(x+y)\geq \tilde{h}(x) + \tilde{h}(y)$ for $x, y\in [0,1]$ with $x+y \leq 1$).

\medskip

\begin{theorem}\label{thm:1}
    Let $h_i\in \mathcal{H}$ for $i\in[n]$ and $h=\bigwedge_{i=1}^n h_i$.
    \begin{enumerate}[(i)]
        \item If $h$ is dually subadditive, then
    \begin{align}\label{ineq:general} \dboxminus_{i=1}^{n}\rho_{h_i}\ge \dboxplus_{i=1}^{n}\rho_{h_i}=\rho_h\geq  \dsquare_{i=1}^{n}\rho_{h_i}.
    \end{align}   
    \item If $h$ is dually subadditive and $h_j=h$ for some $j\in[n]$, then \begin{align}\label{eq:dually}
           \dboxminus_{i=1}^{n}\rho_{h_i}= \dboxplus_{i=1}^{n}\rho_{h_i}=\rho_h\ge   \dsquare_{i=1}^{n} \rho_{h_i}.
        \end{align}
        \item The function
         $h $ is concave if and only if\begin{align}\label{ineq:1} \dboxminus_{i=1}^{n}\rho_{h_i}\geq  \dboxplus_{i=1}^{n}\rho_{h_i}=\rho_h=\dsquare_{i=1}^{n}\rho_{h_i}.
    \end{align}
     
         \item  If  \eqref{eq:dually} holds then $h$ is dually subadditive.
    \end{enumerate}
\end{theorem}

\begin{proof}
(i) It is trivial to see that $\dsquare_{i=1}^{n}\rho_{h_i} \leq \dboxminus_{i=1}^{n}\rho_{h_i}$ and $\dsquare_{i=1}^{n}\rho_{h_i} \leq \dboxplus_{i=1}^{n}\rho_{h_i}$.
The first equality in \eqref{ineq:general} directly follows from \citet[Proposition 5]{embrechts2018quantile}. 
For any distortion functions $h$ and $g$, we have $\rho_h \leq \rho_g$ if $h \leq g$ (see \citet[Lemma 1]{wang2020characterization}). Therefore, for any $h_i\in \mathcal{H}$ such that $h$ is dually subadditive, we have  
    $$
    \dboxminus_{i=1}^{n} \rho_{h_i} \geq \dboxminus_{i=1}^{n} \rho_{h} = \rho_h, 
    $$
which directly follows from the dual subadditivity of $h$ and \citet[Theorem 3]{ghossoub2024counter}. 

\medskip

(ii)  Take $(X_1, \dots, X_n)\in \mathbb{A}_n(X)$ such that $X_j=X$ and $X_i =0$ for $i\in [n]\setminus\{j\}$. 
It is straightforward to verify that this allocation is counter-monotonic.
Then it follows that $$\dsquare_{i=1}^{n} \rho_{h_i}\leq \dboxminus_{i=1}^{n} \rho_{h_i} \leq \rho_h =\dboxplus_{i=1}^{n}\rho_{h_i}.$$ By the equality \eqref{ineq:general}, we can obtain the desired result.

\medskip

(iii) We first show the ``if " part. Letting $h_i=h$ for $i\in [n]$,  we have $\dsquare_{i=1}^{n}\rho_{h}=\rho_h$, which implies that $h$ is concave since $\rho_h$ is subadditive (see \citet[Theorem 3]{wang2020characterization}). 
Next, we show the ``only if " part. 
By \citet[Lemma 1]{wang2020characterization}, we have  
    $$
    \dsquare_{i=1}^{n} \rho_{h_i} \geq \dsquare_{i=1}^{n} \rho_{h} = \rho_h=\dboxplus_{i=1}^{n} \rho_{h_i}. 
    $$
The first equality follows from the comonotonic improvement theorem of \cite{landsberger1994co}; see also \citet[Theorem 1]{ghossoub2024counter}. Moreover, it always holds that $\dsquare_{i=1}^{n} \rho_{h_i} \leq \dboxplus_{i=1}^{n} \rho_{h_i}$. Thus, $\dsquare_{i=1}^{n} \rho_{h_i} = \dboxplus_{i=1}^{n} \rho_{h_i}$. 

\medskip

(iv) The inequality in \eqref{eq:dually} implies that $\dboxminus_{i=1}^{n}\rho_{h}=\rho_h$ by taking $h_i=h$ for $i\in [n]$. Thus, we obtain that $h$ is dually subadditive by \citet[Theorem 2]{ghossoub2024counter}.
\end{proof}

We immediately obtain the following corollary:
 
\begin{corollary}\label{co:1}
    Let $h_i\in \mathcal{H}$ for $i\in[n]$ and $h=\bigwedge_{i=1}^n h_i$. If $h_i$ is concave for $i\in[n]$ and $h_j=h$ for some $j$, then
    \begin{align}
    \label{eq:all-eq}\dboxminus_{i=1}^{n}\rho_{h_i}= \dboxplus_{i=1}^{n}\rho_{h_i}=\rho_h = \dsquare_{i=1}^{n} \rho_{h_i}.
    \end{align}
\end{corollary}

The proof is straightforward since the concavity of $h_i$ for $i\in [n]$ guarantees the concavity of $h$ , thus implying the dual subadditivity of $h$. In this case, the total risk will be absorbed by the least risk-averse agent if such an agent exists within the group. 

\medskip

\begin{example}
If $h$ is concave,   \eqref{eq:all-eq} holds when $X$ is a constant $x\in \R$, without assuming the condition $h_j=h$ for some $j$. 
    Indeed, any constant vector $(x_1,\dots,x_n)$ in $\mathbb A_n(x)$ gives $\sum_{i=1}^n  \rho_{h_i}(x_i)=\rho_h(x)=x$, and the rest follows from Theorem \ref{thm:1} (iii).
\end{example}

\medskip

\begin{example}
    Suppose that all agents have the same risk preferences, i.e., $h_1=\dots=h_n=h$. If $h$ is dually subadditive, then it holds that
    \begin{align*}
           \dboxminus_{i=1}^{n}\rho_{h}= \dboxplus_{i=1}^{n}\rho_{h}=\rho_h.
        \end{align*}
    The converse statement also holds true by using (iv) of Theorem \ref{thm:1}; this also follows from \citet[Theorem 3]{ghossoub2024counter}. 
\end{example}

Based on the results above for risk-averse agents, 
if there is an agent $j$ with the lowest level of risk aversion (that is, $h_j$ is the smallest among $h_1,\dots,h_n$),
then the optimal way to allocate the total risk $X$
is to assign it entirely to this most risk-tolerant agent. The risk preferences of the other agents are irrelevant in this case. However, when heterogeneous risk-seeking agents are involved, the situation is dramatically different; all agents will participate in the gambling and contribute to the aggregate risk. The details will be discussed in Section \ref{sec:risk-seeking}.

\section{Risk-seeking agents}\label{sec:risk-seeking}

This section contains the main technical contributions of this paper, which characterize optimal allocations for the unconstrained and counter-monotonic risk sharing problem in the setting of risk-seeking agents.

\subsection{Main results}
We first note that, when the agents are risk-seeking, we need to constrain the set of allocations to be bounded from below or above, as discussed by \cite{lauzier2024negatively}. 
The next result shows that  the inf-convolution of $\rho_{h_1},\dots,\rho_{h_n}$ for risk-seeking agents is typically negative infinity if 
the set of payoffs $\mathcal{X}$ is taken to be $L^\infty$.

\begin{proposition}
\label{prop:infinite}
Let $\X=L^\infty$.
  If each $h_{i}\in \H, ~i\in [n]$ is convex and is not the identity function, then 
$$
\dsquare_{i=1}^n \rho_{h_i} (X)=\dboxminus_{i=1}^n \rho_{h_i} (X)=-\infty \mbox{~~~~for all  $X\in \X$.}
$$\end{proposition}

\begin{proof}
Let $h=\bigvee_{i=1}^n h_i$. 
It is clear that 
$\dboxminus_{i=1}^n \rho_{h_i} (X)\le \dboxminus_{i=1}^n \rho_{h} (X)$,
for all $X$. We will show  
$\dboxminus_{i=1}^n \rho_{h} (0)=-\infty$.
For a convex distortion function $h_i$, we have $h_i(1/n)<1/n$ if and only if $h_i$ is not the identity function. 
Hence, $h(1/n)=\bigvee _{i=1}^nh_i (1/n)<1/n$, and $h$ is not the identity function. 
 Let $\theta= n h(1/n)<1$.
Take $(A_1,\dots,A_n)\in \Pi_n$ with $\p(A_i)=1/n$ for $i\in [n]$
and $m>0$. 
Define $X_i=nm\id_{A_i}-m$ for $i\in [n]$. Clearly, $(X_1,\dots,X_n)\in \mathbb A^-_n(0)$.  
It follows that $\rho_h(X_i)= nm h(1/n)-m = (\theta -1)m$.
Therefore, 
$$
\dboxminus_{i=1}^n \rho_{h} (0)\le \sum_{i=1}^n \rho_h(X_i) =n(\theta-1)m.
$$
Letting $m\to\infty$ shows that 
$\dboxminus_{i=1}^n \rho_{h} (0)=-\infty$. 
For any $X\in \X$,  using translation invariance and monotonicity of $\rho_h$, and the fact that $X\le \esssup X$, we obtain
$$
\dsquare_{i=1}^n \rho_{h_i} (X)\le \dboxminus_{i=1}^n \rho_{h_i} (X)\le
\dboxminus_{i=1}^n \rho_{h} (X)
\le \dboxminus_{i=1}^n \rho_{h} (\esssup X )
= \dboxminus_{i=1}^n \rho_{h} (0)+\esssup X  =
-\infty,
$$
which completes the proof.  
\end{proof}

Due to Proposition \ref{prop:infinite}, we will take $\X$ to be $L^+$ or $L^-$ in the remainder of the section. These choices correspond to the natural constraint of no short-selling in a financial market. For instance, if $\X=L^+$, then the total loss is nonnegative, and every agent cannot receive a negative loss (that is, a gain) from the allocation. If $\X=L^-$, then the total risk is a gain, and every agent cannot take a loss by sharing the gain. 

Before we proceed, we introduce some additional terminology and notation that will be useful in our further analysis.
Given $n$ distortion functions $h_i \in \mathcal{H}$ for $i=1,\ldots, n$,  their inf-convolution $\dsquare_{i=1}^{n} h_i(x) : [0,1]\mapsto \mathbb{R}$ is defined as
\begin{align*}
    \dsquare_{i=1}^{n} h_i(x)=\inf \left\{ \sum_{i=1}^{n}h_i(x_i): x_i\in [0,1] \ \text{for}\ i\in[n]; ~  \sum_{i=1}^{n}x_i =x\right\}.
\end{align*}
Similarly, the sup-convolution can be defined as 
\begin{align*}
    \lozen_{i=1}^{n} h_i(x)=\sup \left\{ \sum_{i=1}^{n}h_i(x_i):   x_i\in [0,1] \ \text{for}\ i\in[n]; ~  \sum_{i=1}^{n}x_i =x\right\}.
\end{align*}
Here we also introduce the sup-convolution since it will be useful in our analysis and will simplify our result. 
Note that the inf-convolution of real functions $h_1,\dots,h_n$ is similar to the inf-convolution of risk measures, but the domain is restricted to real numbers in $[0,1]$.

\begin{theorem}\label{thm:convex} 
     Suppose that  $h_{i}\in \mathcal{H} $ is continuous and convex  for $i\in[n]$, and that $X\in \mathcal{X}^{\perp}$ with $\X=L^+$ or $\X=L^-$. Then
     \begin{align}\label{eq:convex1}
        \dsquare_{i=1}^n \rho_{h_{i}}(X) = \dboxminus_{i=1}^n \rho_{h_{i}}(X) = \rho_{g}(X), \end{align}
        where  $g$ is such that (i) $g=\dsquare_{i=1}^n h_{i}$ if $\mathcal X=L^+$; (ii)
       $\tilde g=\lozen_{i=1}^n \tilde{h}_i$ if $\mathcal X=L^-$. 
\end{theorem}
The proof of Theorem \ref{thm:convex} involves three lemmas and is presented in Section \ref{sec:pf-th3}.     
For the case where all agents are risk averse, the relationship among the three types of inf-convolutions can be established using Theorem \ref{thm:1}(iii), implying that agents always prefer comonotonic allocations. In contrast, when all agents are risk seeking, Theorem \ref{thm:convex} indicates that the relationship is determined by a preference for counter-monotonic allocations. Specifically, for $\mathcal{X}=L^+$ (resp.~$\mathcal{X}=L^-$) and $X\in \mathcal{X}^{\perp}$, we have 
\begin{align}\label{ineq:2}
    \dboxplus_{i=1}^n \rho_{h_{i}}(X)\geq \dsquare_{i=1}^n \rho_{h_{i}}(X) = \dboxminus_{i=1}^n \rho_{h_{i}}(X).
\end{align}
Notably, the comonotonic inf-convolution is always a distortion risk measure, whereas the inf-convolution of convex distortion risk measures is no longer a distortion risk measure, but rather a monotone distortion riskmetric, as stated in Theorem \ref{thm:convex}. To further verify \eqref{ineq:2} for risk-seeking agents, we provide several numerical experiments, reported in Table \ref{tab:1}. In these examples, we consider two agents with $h_1(x)=1-(1-x)^{0.3}$ and $h_2(x)=1-(1-x)^{0.6}$, respectively. Clearly, both agents are risk seeking.

\begin{table}[t]
\renewcommand{\arraystretch}{1.5}
\centering
\begin{tabular}{c|c|c|c|cc} 
 & $X$ & $\X$ &$ \dboxplus_{i=1}^2 \rho_{h_i}(X)$ &  $\dboxminus_{i=1}^2 \rho_{h_i}(X)= \dsquare_{i=1}^2 \rho_{h_i}(X)$\\ 
\hline
\multirow{2}{*}{$Y \sim \text{Uniform}(0,1)$ } &$Y$ & $L^+$  & 0.3692 & 0.2074  \\
& $-Y$ & $L^-$ & -0.7667 & -1.0435  \\
\hline
\multirow{2}{*}{$Y \sim \text{Pareto}(3,2)$} &$Y$ & $L^+$   & 2.7291 & 1.3828 \\
& $-Y$  & $L^-$ & -9.4262 & -11.0881 \\
\hline
\multirow{2}{*}{$Y \sim \text{logN}(0,1)$} &$Y$ & $L^+$  & 1.0825 & 0.5849 \\
& $-Y$  & $L^-$ & -5.5828 & -6.4773 \\
\hline
\end{tabular}
\label{table:3}
\caption{Comparison of the three inf-convolutions.}\label{tab:1}
\end{table}

\begin{example}
    When all agents have the same risk preference, that is, $h_1=\dots=h_n=h$, and when $h$ is convex and continuous on $[0,1]$, Theorem \ref{thm:convex} reduces to the homogeneous case, examined in \citet[Theorem 3]{ghossoub2024counter}. For example, when $\mathcal{X}=L^+$, the function $g$ becomes 
$$
g(x)=\dsquare_{i=1}^n h(x) = n h(x/n),  ~~ x\in [0,1], 
$$
which follows from the convexity of $h$. Similarly, when $\mathcal{X}=L^-$, we have 
$$
g(x)= \tilde{l}(x),\ \text{and}\ l(x)=\lozen_{i=1}^n \tilde{h}(x) =n \tilde{h}(x/n), ~~ x\in [0,1],
$$
and it is straightforward to verify that $g(x)=nh(1-(1-x)/n)-nh(1-1/n)$.
\end{example}

\subsection{Three technical lemmas and the proof of Theorem \ref{thm:convex}}
\label{sec:pf-th3}
We first present a technical lemma.

\begin{lemma}\label{lem:new} 
For a given $x>0$, any random variable $X $, and $h\in \mathcal H$,
we have 
$$
 \inf \left\{\sum_{i=1}^{n} h(\mathbb{P}(X\mathbb{1}_{A_{i}} \geq x)): (A_{1}, \ldots, A_{n})\in \Pi_{n}  \right\}
 = \dsquare_{i=1}^n h (\p(X\ge x)).
$$   
\end{lemma}

\begin{proof}
Let $p=\p(X\ge x)$. If $p=0$, then there is nothing to show. Suppose $p>0$ in what follows. 
    Since the probability space $(\Omega,\mathcal F,\p)$  is atomless, 
    for any $(p_1,\dots,p_n)\in \R_+^n$ with $p_1+\dots+p_n=p$, there exists a composition $(A_1,\dots,A_n) $ of 
    $\Omega$ such that 
    $\p(A_i\cap \{X\ge x\})=p_i$ for $i\in [n]$.
    Therefore, 
$$
 \inf \left\{\sum_{i=1}^{n} h(\mathbb{P}(X\mathbb{1}_{A_{i}} \geq x)): (A_{1}, \ldots, A_{n})\in \Pi_{n}  \right\}
 \ge  \dsquare_{i=1}^n h (\p(X\ge x)).
$$ 
The converse direction is straightforward since 
$\sum_{i=1}^n \mathbb{P}(X\mathbb{1}_{A_{i}} \geq x)=\p(X\ge x) = p.$
\end{proof}

\begin{lemma}\label{le:1}
    If for each $i\in [n]$, the function $h_{i}\in \H$ is strictly convex and differentiable, then there exist increasing functions $f_{i}: [0,1]\mapsto [0,1]$, $i=1,\dots,n$, such that $\dsquare_{i=1}^{n} h_{i}(x)=\sum_{i=1}^{n} h_{i}(f_{i}(x))$, where $\sum_{i=1}^{n}f_{i}(x)=x$ and $x\in [0,1]$.
\end{lemma}

\begin{proof}
For a fixed $x\in [0,1]$, our aim is to find the infimum of $g: (x_1, \dots, x_n)\mapsto \sum_{i=1}^{n} h_i(x_i)$ over $(x_1, \dots, x_n)\in [0,1]^n$, subject to the constraint $x_1 + \dots + x_n = x$.
It is straightforward to see that $g$ is strictly convex as  $h_{i}$ is strictly convex for all $i\in [n]$. Thus, the minimum of $g(x_1,\dots,x_n)$
is attained at, say, $(x_1^{*}, \dots, x_n^{*})$. 
Write  $$ g^* (x_1,\dots,x_{n-1})= 
g\left(x_1,\dots,x_{n-1}, x-\sum_{i=1}^{n-1} x_i\right)$$
for $(x_1,\dots,x_{n-1})\in [0,1]^n$ with  $\sum_{i=1}^{n-1} x_i\le x.$  The first-order condition implies  \begin{align*} 
\frac{\partial g^*}{\partial x_i}
(x^*_1,\dots,x^*_{n-1})= h_{i}^{\prime}(x_{i}^{*}) - h_{n}^{\prime}\left(x-\sum_{j=1}^{n-1}x_{j}^{*}\right)=0, ~~~ i\in [n-1], \end{align*}
which implies
\begin{equation}
\label{eq:R1-1}
h_{i}^{\prime}(x_{i}^{*})= h_{n}^{\prime} (x_n^* ),~~~i\in [n-1].
\end{equation} 
Write 
$t_i=(h_i')^{-1}$ for $i\in [n]$.  Using \eqref{eq:R1-1}, we get 
$$
\sum_{i=1}^{n}t_i (h_{n}^{\prime}(x_{n}^{*}))= 
\sum_{i=1}^{n}t_i (h_{i}^{\prime}(x_{i}^{*})) =\sum_{i=1}^n x_i^*= x. 
$$
Therefore, 
$$
h_{i}^{\prime}(x_{i}^{*}) = h_{n}^{\prime}(x_{n}^{*})  = \left(\sum_{i=1}^{n}t_i\right)^{-1}(x) ~~ \text{\ \ \ for all}\  i\in [n].
$$
Consequently, the value of $x_{i}^{*}$ can be represented as $x_{i}^{*}= t_i ((\sum_{i=1}^{n}t_i)^{-1}(x))$. 
Furthermore, from the strict convexity of $h_1,\dots,h_n$,  the functions  $ t_1,\dots,t_n$, as well as $\sum_{i=1}^{n} t_i $,  are increasing. Therefore, we obtain the desired result by taking the increasing functions $f_1,\dots,f_n$ specified by  $f_{i}(y)=t_i ((\sum_{i=1}^{n}t_i )^{-1}(y))$ for $y \in [0,1]$ and $i\in [n]$.
\end{proof}

In what follows, $C^2[0,1]$ is the set of all continuous functions on $[0,1]$ 
with continuous second-order derivatives on $(0,1)$.

\begin{lemma}\label{le:2}
    Suppose that $h_{i}\in \mathcal{H} \cap C^2[0,1]$ is strictly convex for $i\in[n]$, and that $X\in \mathcal{X}^{\perp}$. 
    \begin{itemize}
        \item[(i)] If $\mathcal X=L^+$, then
         $
        \dsquare_{i=1}^n \rho_{h_{i}}(X) = \dboxminus_{i=1}^n \rho_{h_{i}}(X) = \rho_{g}(X), $
        where $g=\dsquare_{i=1}^n h_{i}$.
        \item[(ii)] If $\mathcal X=L^-$, then
        $
        \dsquare_{i=1}^n \rho_{h_{i}}(X) = \dboxminus_{i=1}^n \rho_{h_{i}}(X) = \rho_{g}(X),
        $
    where $g$ is such that $\tilde g=\lozen_{i=1}^n \tilde{h}_i$.
    \end{itemize}

\end{lemma}

\begin{proof}
    (i) 
    Suppose that $\mathcal X=L^+$ and $X\in \mathcal X^{\perp}$. By Theorem \ref{theorem:counter_impro}, for any allocation $(X_{1}, \ldots, X_{n})\in \mathbb{A}_{n}(X)$,  
    there exists a jackpot allocation $(X\mathbb{1}_{A_{1}}, \ldots, X\mathbb{1}_{A_{n}} )$, $(A_{1}, \ldots, A_{n})\in \Pi_{n}$ such that $X\mathbb{1}_{A_{i}}\leq_{\mathrm{cv}} X_i$ for each $i\in [n]$.  Since $h_{i}\in \mathcal{H}$ for $i\in [n]$ are convex, there exists $(A_{1}, \ldots, A_{n})\in \Pi_{n}$ such that 
    \begin{align*}
        \sum_{i=1}^{n} \rho_{h_{i}}(X\mathbb{1}_{A_{i}}) = 
        \sum_{i=1}^{n} \int_{0}^{\infty} h_{i}(\mathbb{P}(X\mathbb{1}_{A_{i}} \geq x)) \mathrm{d} x 
        \leq
        \sum_{i=1}^{n} \rho_{h_{i}}(X_{i})
    \end{align*}
    holds for $(X_{1}, \ldots, X_{n})\in \mathbb{A}_{n}(X)$.
   Taking the infimum on both sides yields 
   \begin{align*}
       \inf \left\{\int_{0}^{\infty}\sum_{i=1}^{n} h_{i}(\mathbb{P}(X\mathbb{1}_{A_{i}} \geq x)) \mathrm{d} x: (A_{1}, \ldots, A_{n})\in \Pi_{n} \right\}
       \leq \dsquare_{i=1}^n \rho_{h_{i}}(X) \leq \dboxminus_{i=1}^n \rho_{h_{i}}(X),
   \end{align*}
and the above inequalities are in fact equalities since 
$(X\mathbb{1}_{A_{1}}, \ldots, X\mathbb{1}_{A_{n}} )\in \mathbb A^-_n(X)$. Now, let $g=\dsquare_{i=1}^n h_{i}$. 
 Using Lemma \ref{lem:new}, 
 the above inequalities imply that
   \begin{align*}
       \rho_{g}(X)=\!
       \int_{0}^{\infty}\! \inf \left\{\sum_{i=1}^{n} h_{i}(\mathbb{P}(X\mathbb{1}_{A_{i}} \geq x)):\! (A_{1},\! \ldots,\! A_{n})\in \Pi_{n}  \right\} \mathrm{d} x
       \leq \dsquare_{i=1}^n \rho_{h_{i}}(X) \leq \dboxminus_{i=1}^n \rho_{h_{i}}(X).
   \end{align*}
  Thus, we obtain that $\rho_{g}(X) \leq \dboxminus_{i=1}^n \rho_{h_{i}}(X)$.

    Next, we show the converse direction, that is, $\rho_g(X)$ is attainable by some allocation. Since $h_{i}$ is strictly convex and differentiable for each $i\in[n]$, there exist increasing functions $f_{i}$, $i\in[n]$, such that $\dsquare_{i=1}^{n}h_{i}(x)=\sum_{i=1}^{n}h_{i}(f_{i}(x))$ and $\sum_{i=1}^{n}f_{i}(x)=x$, $x\in [0,1]$, as stated in Lemma \ref{le:1}. Define the events $A_1,\dots,A_n$ by
 \begin{align*}
     A_{i}&=\left\{1-\sum_{j=1}^{i}f_{j}^{\prime}(1-U_X)<U\leq 1-\sum_{j=1}^{i-1}f_{j}^{\prime}(1-U_X)\right\}, ~~~ i\in [n-1],\\ 
     A_{n}&=\left\{U\leq 1-\sum_{j=1}^{n-1}f_{j}^{\prime}(1-U_X)\right\}.
 \end{align*}
 Note that the strict convexity of $h_i\in C^2[0,1]$ guarantees the existence of $f_i'$ on $(0,1)$.
By the above construction, it is straightforward to verify that $(A_{1}, \dots, A_{n})$ is a composition of $\Omega$, since $\sum_{i=1}^{n}f_{i}^{\prime}(x) =1$ for $x \in(0, 1)$. Moreover, for $i\in [n]$ and $x>0$,
$\p(A_i|X\ge x) = \E[f'_i(1-U_X)|X\ge x]$, due to independence between $U$ and $X$. Consider the allocation $(X_{1}, \dots,X_{n})=(X\mathbb{1}_{A_{1}}, \dots, X\mathbb{1}_{A_{n}})$, which is in  $\mathbb{A}_{n}^{-}(X)$.
Note that $f_i(y)\to 0$ for $y\downarrow 0$ because $f_i(x)\le x$ for $x\in (0,1)$. Then 
\begin{align}\notag
    \sum_{i=1}^{n} \rho_{h_{i}}(X\mathbb{1}_{A_{i}})=&\sum_{i=1}^{n} \int_{0}^{\infty}h_{i}(\mathbb{P}(X\mathbb{1}_{A_{i}} \geq x)) \mathrm{d} x\\ \notag
 =& \sum_{i=1}^{n} \int_{0}^{\infty}h_{i}(\mathbb{P}(A_{i}|X \geq x)\mathbb{P}(X \geq x)) \mathrm{d} x\\ \notag
 =& \sum_{i=1}^{n} \int_{0}^{\infty}h_{i}\left(\mathbb{E}[f_{i}^{\prime}(1-U_{X})|X\geq x]\mathbb{P}(X \geq x)\right)\mathrm{d} x\\ \notag
  =& \sum_{i=1}^{n} \int_{0}^{\infty}h_{i}\left(\mathbb{E}\left[f_{i}^{\prime}(1-U_{X})\mathbb{1}_{\left\{F_X^{-1}(U_X)\geq x\right\}}\right]\right)\mathrm{d} x\\ \label{eq:11}
=& \sum_{i=1}^{n} \int_{0}^{\infty}h_{i}\left(\int_{\mathbb{P}(X\leq x)}^{1} f_{i}^{\prime}(1-t)\d t\right)\mathrm{d} x\\ \notag
=& \sum_{i=1}^{n} \int_{0}^{\infty}h_{i}\left( f_{i}(S_X(x))\right) \mathrm{d} x\\ \notag
=& \int_{0}^{\infty} \dsquare_{i=1}^{n}h_{i}(S_X(x)) \mathrm{d} x = \rho_{g}(X).
\end{align}

The equality \eqref{eq:11} holds due to the equivalence of  $F_X^{-1}(U_X)\geq x$ and $U_X \geq \mathbb{P}(X\leq x)$ (see \citet[Lemma 1]{guan2024reverse}).
Hence, the result implies that $\dboxminus_{i=1}^{n}\rho_{h_{i}}(X)\leq \rho_{g}(X)$. Combining the above, we obtain $\dboxminus_{i=1}^{n}\rho_{h_{i}}(X)= \rho_{g}(X)$.

\medskip

   (ii) This part follows by symmetric arguments to part (i). 
   Denote by $\ell_{i}(x)=-\tilde{h}_{i}(x)$ and $h^*(x)=\dsquare_{i=1}^n \ell_{i}(x)$. 
   By straightforward calculations, we have 
   \begin{align*}
       \dsquare_{i=1}^{n}\rho_{h_i}(X) &=\inf \left\{\sum_{i=1}^{n}\rho_{h_i}(X_i), (X_1, \cdots, X_n)\in \mathbb{A}_n(X) \right\} 
       \\ &=
       \inf \left\{\sum_{i=1}^{n}\rho_{\ell_{i}}(-X_i), (-X_1, \cdots, -X_n)\in \mathbb{A}_n(-X) \right\}\\
       &= \rho_{h^*}(-X)=\rho_{g}(X),
   \end{align*}
   where $g(x)=h^*(1-x)-h^*(1)=\dsquare_{i=1}^n \ell_{i}(1-x)-\dsquare_{i=1}^n \ell_{i}(1)=
       \lozen_{i=1}^n \tilde{h}_{i}(1)-
       \lozen_{i=1}^n \tilde{h}_{i}(1-x).$   
       The last but one equality follows from the case of $\mathcal{X}=L^{+}$.   
\end{proof}

\begin{proof}[Proof of Theorem \ref{thm:convex}]
The main idea of the proof is to approximate a general function $h_i$ by its Bernstein polynomial $B_k^i$ for $k\ge 1$, which is twice differentiable, and then to apply Lemma \ref{le:2}. 
We only prove case (i).  Fix $i\in [n]$ for now. 
    For the continuous function $h_i $ on $[0,1]$, and for $\epsilon>0$, there exists an integer $N_i\geq 2$ such that   for $k>N_i$, \begin{align}\label{ineq:3}
       \sup_{x\in [0,1]} |h_i(x)-B_k^i(x)|<\epsilon,\  \text{where}\  B_k^i( x)=\sum_{r=0}^k h_i \left(\frac{r}{k}\right)\binom{k}{r} x^r(1-x)^{k-r},\  x\in [0,1];
    \end{align}
   see \citet[Theorem 7.1.5]{phillips2003interpolation}.
    Clearly, for all $k\geq 1$, $B_k^i(0) = h_i(0)=0$ and $B_k^i(1) = h_i(1)=1$. Also, since $h_i$ is increasing and convex,   $B_k^i$ is also increasing and convex (see \citet[Theorem 7.1.4]{phillips2003interpolation}).
Let $\tilde B_k^i
(x)= (1-\epsilon)B_k^i(x) + \epsilon x^2
$ for $x\in [0,1]$ and $k\ge 1$. It is clear that $\tilde B_k^i$ is a strictly convex and increasing function 
such that  
$|\tilde B_k^i-B^i_k|\le \epsilon$,  
$\tilde B_k^i(0)=0$
and $\tilde B_k^i(1)=1$.

    Since $X$ is bounded, we assume that $X\leq M$, for some constant $M\geq 0$.
    Fix $k>\max\{N_1,\dots,N_n\}$. It holds that for all $Y$ supported on $[0,M]$,
     \begin{align*}
        |\rho_{h_i}(Y)-\rho_{B_k^i}(Y)|  \leq 
        \int_{0}^{M}\left| h_{i}(\mathbb{P}(Y \geq x))- B_k^i (\mathbb{P}(Y \geq x))\right| \mathrm{d} x
        <M\epsilon,
    \end{align*} 
    for all $i\in [n]$, and hence,
    $$
    |\rho_{h_i}(Y)-\rho_{\tilde B_k^i}(Y)|  \leq 2M\epsilon ~~\text{and}~~
    \left|\dboxminus_{i=1}^{n} \rho_{h_i}(Y) -\dboxminus_{i=1}^{n} \rho_{\tilde B_k^i} (Y)\right|<2 n M\epsilon.
    $$ 
    Note that  $\tilde B_k^i \in \mathcal H\cap C^2[0,1]$ for $i\in [n]$. It follows that $\dboxminus_{i=1}^{n} \rho_{\tilde B_k^i}(X) =\rho_{g_k}(X) $, where $g_k=\dsquare_{i=1}^{n}\tilde B_k^i$, by 
    Lemma \ref{le:2}. 
Using the inequality \eqref{ineq:3} again, and writing $h=\dsquare_{i=1}^{n}h_i$, we have
    \begin{align*}
    \left\Vert h-g_k \right\Vert=\left\Vert\dsquare_{i=1}^{n}h_i-\dsquare_{i=1}^{n}\tilde B_k^i\right\Vert<2 n\epsilon,
    \end{align*}  
    where $\Vert\cdot\Vert$ is the supremum norm for continuous functions on $[0,1]$.
    Because $\mathcal X=L^+$ and  $X\in \mathcal{X}^{\perp}$, it follows that
    \begin{align*}
        \left|\dboxminus_{i=1}^{n} \rho_{h_i}(X) -\rho_h(X) \right|&=\left|\dboxminus_{i=1}^{n} \rho_{h_i}(X) -\rho_{g_k}(X) +\rho_{g_k}(X) -\rho_h(X) \right|\\
        &\leq \left|\dboxminus_{i=1}^{n} \rho_{h_i}(X) -\rho_{g_k}(X) \right|+\left|\rho_{g_k}(X)-\rho_h(X)\right|\le 4nM\epsilon.
    \end{align*}
    Thus, (i) holds true.
    The proof of (ii) follows from noticing that $-X$ satisfies the assumptions of (i). 
\end{proof}

\subsection{Sharing a constant payoff}
 
We take a closer look at the case $\X=L^+$ and when the aggregate risk $X$ is a constant. By \eqref{eq:convex1} and the positive homogeneity of distortion risk metrics, we set $X=1$ without loss of generality. In this case, Theorem \ref{thm:convex} implies that a class of Pareto-optimal allocations takes the form of  $X_i=\mathbb{1}_{A_{i}}$ for $i\in [n]$, where $(A_1,\dots, A_n)\in \Pi_n$  satisfies $\dsquare_{i=1}^{n}h_i(1)=\sum_{i=1}^{n}h_i(\mathbb{P}(A_i))$. As we can see, the optimal partition depends entirely on the agents' distortion functions $h_i$, which capture their individual risk preferences. Thus, in a setting with no aggregate uncertainty, the optimal allocation assigns to each risk-seeking agent $i$ a probability $\mathbb{P}(A_i)$ of bearing the total loss. This highlights a key economic feature of risk sharing among risk-seeking agents: optimal allocations concentrate the loss rather than diversify it, and each agent receives a ``probability share'' of the full loss according to the structure of the solution. Essentially, the problem boils down to a ``probability sharing'' problem among agents with different attitudes toward risk. It is therefore natural to conjecture that an efficient sharing rule tends to assign higher probabilities of bearing the loss to agents with stronger risk–seeking preferences. However, as shown in Proposition \ref{prop:power}, such monotonicity can fail in the case of two agents. The following result focuses on agents whose preferences are represented by convex power distortions, and it illustrates how the agents' attitudes toward risk shape the structure of an optimal allocation within this particular class.

\begin{proposition}\label{prop:power}
Let $n \in \mathbb{N}$,  
$\alpha_1 \leq \dots \leq \alpha_n \in (1, \infty)$, 
and  $h_i(x)=x^{\alpha_i}$ for $x\in [0,1]$.  Suppose that $(x_1, \cdots, x_n)\in [0,1]^n$ satisfies $\sum_{i=1}^{n} x_i=1$ and $\sum_{i=1}^{n}h_i(x_i)=\dsquare_{i=1}^{n} h_i(1)$. 
\begin{itemize}
    \item[(i)] In case $n=2$, we have 
    $x_1\leq x_2 $ if and only if $ {\alpha_2}2^{-\alpha_2}\leq {\alpha_1}2^{-\alpha_1}$.
\item[(ii)] In case $n\geq 3$, we have  $x_1 \leq \dots \leq x_n$.
\end{itemize}
\end{proposition}

\begin{proof}    
    Our goal is to find $x_1,\dots, x_n$ such that $\sum_{i=1}^{n}x_i=1$ and $\dsquare_{i=1}^{n}h_i(1)=\sum_{i=1}^{n}h_i(x_i)$; that is, to solve
    \begin{align}
    \label{eq:solve-w}
        \argmin_{w_1, \dots, w_{n}} \left\{\sum_{i=1}^{n}(w_i)^{\alpha_i}: w_1, \dots, w_{n}\geq 0\ \text{and}\  \sum_{i=1}^{n}w_i = 1\right\}.
    \end{align}

    First, we consider the case of $n=2$. The optimal share $(x_1, x_2)\in [0,1]^2$
    can be solved by the first-order condition of $\alpha_1 {x_1}^{\alpha_1-1}-\alpha_2 (1-x_1)^{\alpha_2-1}=0$. Let $m=\frac{x_1}{1-x_1}$. The condition can be written as:
    $$
    m^{\alpha_1-1}(1+m)^{\alpha_2-\alpha_1}=\frac{\alpha_2}{\alpha_1}.
    $$
    Clearly, $F(y):=y^{\alpha_1-1}(1+y)^{\alpha_2-\alpha_1}$ is non-decreasing in $y$ and $F(1)=2^{\alpha_2-\alpha_1}$. Therefore, the condition of $ {\alpha_2}/{\alpha_1}\leq F(1)$ is equivalent with $m\leq 1$, implying that $x_1\leq 1/2$ and $x_1 \leq x_2$.
    
    Next, we consider the case of $n\geq 3$.
    To solve the problem, we define the Lagrangian as
    \begin{align}\label{eq:la}
        \mathcal{L}(w_1, \dots, w_n;\lambda)=\sum_{i=1}^{n}w_i^{\alpha_i} + \lambda \left(1-\sum_{i=1}^{n}w_i\right),\  \lambda\in \mathbb{R}.
    \end{align}
    The first-order condition for \eqref{eq:la} is given by 
    \begin{align}\label{eq:4}
        \alpha_1 {x_1}^{\alpha_1-1}= \dots =\alpha_n {x_n}^{\alpha_n-1} = \lambda,\ \text{and}\ \sum_{i=1}^{n}x_i = 1.
    \end{align}
    Thus, it follows that
    \begin{align*} 
        f(\lambda):=\sum_{i=1}^{n} \left(\frac{\lambda}{\alpha_i}\right)^{\frac{1}{\alpha_i-1}}=1.
    \end{align*}
    Let $$y_i=\left(1+\frac{1}{\beta_i}\right)^{-\beta_i}~\text{and}~\beta_i=\frac{1}{\alpha_i-1}.$$
    Following the fact that $g(x):=(1+1/x)^x$ is increasing on $(0,\infty)$ with range $(1,e)$, we obtain that $y_i$ is decreasing in $\beta_i$ and satisfies $y_i>1/e$. Therefore,  $$\sum_{i=1}^{n} y_i=f(1)=\sum_{i=1}^{n} \left(\frac{1}{\alpha_i}\right)^{\frac{1}{\alpha_i-1}}> \frac{n}{e} \ge \frac{3}{e}>1.$$
    Since $f(\lambda)$ is strictly increasing with $\lambda$, the solution of   $f(\lambda)=1$ satisfies $\lambda<1$.
    Consequently, it follows that $x_i=\lambda ^{\beta_i}y_i$ is decreasing in $\beta_i$, and hence increasing in $\alpha_i$.
\end{proof}

Proposition \ref{prop:power} shows that when $n\geq3$, optimal allocations always assign higher probabilities of bearing the loss to agents with more risk-seeking distortion functions. The ordering of probabilities then aligns naturally with the ordering of $\alpha_i$. However, this monotonicity does not necessarily hold in the two-agent case.
    When $n=2$, the feasibility constraint 
    $x_1+x_2=1$ restricts the set of admissible allocations to be a one-dimensional curve.
    The optimal allocation is therefore determined by comparing the agents' marginal distortion penalties at the symmetric point $x_1=x_2=1/2$.
    As a result, the ordering of $(x_1,x_2)$ is governed not directly by the risk-seeking parameters $(\alpha_1, \alpha_2)$, but by the transformed quantities $(\alpha_1 2^{-\alpha_1},\alpha_2 2^{-\alpha_2})$.
    Since the mapping $x\mapsto x 2^{-x}$
    is not monotone on $(1,\infty)$, it is possible for the efficient allocation to assign a larger probability of bearing the loss to the less risk-seeking agent. 
    For example, when $\alpha_2 2^{-\alpha_2}>\alpha_1 2^{-\alpha_1}$, as in the case $h_1(x)=x^{1.2}$ and $h_2(x)=x^{1.4}$, the efficient solution assigns a larger probability of loss to the less risk-seeking agent, as shown in the Table \ref{tab:2}.  In contrast, when the ordering reverses, such as the pair of $h_1(x)=x^{1.2}$ and $h_2(x)=x^5$, the efficient allocation assigns a larger probability to the more risk-seeking agent.
    This counterintuitive outcome arises from a structural feature of the two-agent setting: with only one degree of freedom, increasing one agent's share automatically decreases the other's. 
    Consequently, the exact curvature of the distortion functions near the symmetric split becomes more influential than the overall level of risk seeking, and the less risk-seeking agent may end up bearing more of the loss.

\begin{table}[ht]\vspace{0.2cm}
\centering
\renewcommand{\arraystretch}{1.5}  
\begin{tabular}{cc|ccc}
\hline
$h_1$ & $h_2$ & $\mathbb{P}(A_1)$& $\mathbb{P}(A_2)$ & $\rho_{h_1}\dsquare \rho_{h_2}(1)$ \\ \hline
 $x^{1.2}$    &   $x^{1.4}$    &  0.5129       &  0.4871 &   0.8141   \\ \hline
  $x^{1.2}$   & $x^{5}$     &    0.3371     &   0.6629      &    0.3992  \\ \hline
\end{tabular}
\caption{Comparison of the inf-convolution and the optimal allocation with different $h_1$ and $h_2$.}\label{tab:2}
\end{table}

\section{Conclusion}\label{sec:con}
The counter-monotonic improvement theorem (reported as Theorem \ref{theorem:counter_impro} herein) not only provides insights into solving risk sharing problem among non-risk-averse agents, but also serves as the foundation for studying a counter-monotonic risk exchange mechanism, in which only counter-monotonic risk allocations are allowed. In this paper, we investigate the counter-monotonic risk sharing problem for heterogeneous underlying risk measures $\rho_{h_1}, \dots, \rho_{h_n}$.

When all agents are risk averse, meaning that $h_i$ is concave for each $i \in [n]$, the following relationship holds (Theorem \ref{thm:1} and Corollary \ref{co:1}): 
\begin{align*}
\dboxminus_{i=1}^{n}\rho_{h_i}\geq  \dboxplus_{i=1}^{n}\rho_{h_i}=\rho_h=\dsquare_{i=1}^{n}\rho_{h_i},\ \text{where}\ h=\bigwedge_{i=1}^{n}h_i.
\end{align*}
In this case, comonotonic allocations are always preferred, implying that the counter-monotonic inf-convolution generally yields a larger value than the other two. As a result, finding a closed-form characterization of the counter-monotonic inf-convolution becomes challenging, and it is not our focus in this paper. However, we provide a sufficient condition under which $\dboxminus_{i=1}^{n}\rho_{h_i}$ is equal to the other two inf-convolutions.

When $h_i$ is convex and continuous on $[0,1]$ for each $i \in [n]$, indicating that agents are risk seeking, the inf-convolution for such distortion risk measures admits explicit formulas for $X\in \mathcal{X}^\perp$: 
\begin{align*}
\dsquare_{i=1}^n \rho_{h_{i}}(X) = \dboxminus_{i=1}^n \rho_{h_{i}}(X) = \rho_{g}(X), 
\end{align*}
where  $g$ is such that (i) $g=\dsquare_{i=1}^n h_{i}$ if $\mathcal X=L^+$; and (ii) $\tilde g=\lozen_{i=1}^n \tilde{h}_i$ if $\mathcal X=L^-$. In this setting, these results show that a representative agent (defined by the inf-convolution as its reference) of several agents with distortion risk measures is no longer an agent with a distortion risk measure, but rather with a distortion riskmetric.

\section*{Acknowledgements}
 Mario Ghossoub acknowledges financial support from the Natural Sciences and Engineering Research Council of Canada (RGPIN-2024-03744). Ruodu Wang acknowledges financial support from the Natural Sciences and Engineering Research Council of Canada (RGPIN-2024-03728, CRC-2022-00141). 

\bibliography{ref}

\end{document}